\documentclass[12pt]{article}

\usepackage{verbatim,color,amsmath,amsthm,amssymb,latexsym}
\usepackage[usenames,dvipsnames]{xcolor}
\usepackage{fancyhdr}
\usepackage[authoryear, sort]{natbib}
\usepackage{graphicx, epstopdf, subcaption}
\usepackage{booktabs}
\usepackage{mathrsfs}
\usepackage{etoolbox}

\newcommand{\sbf}{\boldsymbol}

\newtheorem{thm}{\underline{\bf Theorem}}

\newtheorem{remark}{\underline{\bf Remark}}
\newtheorem{proposition}{Proposition}

\newtheorem{lemma}{\underline{\bf Lemma}}

\def\bse{\begin{eqnarray*}}
\def\ese{\end{eqnarray*}}
\def\be{\begin{eqnarray}}
\def\ee{\end{eqnarray}}
\def\bq{\begin{equation}}
\def\eq{\end{equation}}
\def\bse{\begin{eqnarray*}}
\def\ese{\end{eqnarray*}}

\def\wt{\widetilde}

\def\boxit#1{\vbox{\hrule\hbox{\vrule\kern6pt  \vbox{\kern6pt#1\kern6pt}\kern6pt\vrule}\hrule}}
\makeatletter
\newcommand*\rel@kern[1]{\kern#1\dimexpr\macc@kerna}
\newcommand*\widebar[1]{%
  \begingroup
  \def\mathaccent##1##2{%
    \rel@kern{0.8}%
    \overline{\rel@kern{-0.8}\macc@nucleus\rel@kern{0.2}}%
    \rel@kern{-0.2}%
  }%
  \macc@depth\@ne
  \let\math@bgroup\@empty \let\math@egroup\macc@set@skewchar
  \mathsurround\z@ \frozen@everymath{\mathgroup\macc@group\relax}%
  \macc@set@skewchar\relax
  \let\mathaccentV\macc@nested@a
  \macc@nested@a\relax111{#1}%
  \endgroup
}
\makeatother

\begin{document}

\thispagestyle{empty}
\baselineskip=28pt
{\LARGE{\bf A Divide and Conquer Algorithm of Bayesian Density Estimation}}

\baselineskip=12pt

\vskip 2mm
\begin{center}
Ya Su \\
Department of Statistics, University of Kentucky, Lexington, KY 40536-0082, U.S.A., ya.su@uky.edu\\
\hskip 5mm \\
\end{center}

\begin{center}
{\Large{\bf Abstract}}
\end{center}
\baselineskip=12pt Data sets for statistical analysis become extremely large even with some difficulty of being stored on one single machine. Even when the data can be stored in one machine, the computational cost would still be intimidating. We propose a divide and conquer solution to density estimation using Bayesian mixture modeling including the infinite mixture case. The methodology can be generalized to other application problems where a Bayesian mixture model is adopted. The proposed prior on each machine or subsample modifies the original prior on both mixing probabilities as well as on the rest of parameters in the distributions being mixed. The ultimate estimator is obtained by taking the average of the posterior samples corresponding to the proposed prior on each subset. Despite the tremendous reduction in time thanks to data splitting, the posterior contraction rate of the proposed estimator stays the same (up to a $\log$ factor) as that of the original prior when the data is analyzed as a whole. Simulation studies also justify the competency of the proposed method compared to the established WASP estimator in the finite dimension case. In addition, one of our simulations is performed in a shape constrained deconvolution context and reveals promising results. The application to a GWAS data set reveals the advantage over a naive method that uses the original prior.

\baselineskip=12pt
\par\vfill\noindent
\underline{\bf Some Key Words}: Divide and conquer, Bayesian density
estimation, Posterior contraction rate, Bayesian mixture model. 

\par\medskip\noindent
\underline{\bf Short title}:

\clearpage\pagebreak\newpage
\pagenumbering{arabic}
\newlength{\gnat}
\setlength{\gnat}{22pt}
\baselineskip=\gnat

\section{Introduction}

In an era of real big data, data sets
for statistical analysis become extremely large even with some difficulty
of storing on one single machine. Even when the data can be stored in
one machine, the computational cost would still be intimidating. As an
example, for most recent data sets in the Genome Wide
Association Study (GWAS), the number of subjects amounts to several hundreds of
thousands while the number of single-nucleotide polymorphism (SNP)
goes up to one million for
each individual.
       
A divide and conquer algorithm involves three steps. First, a
partition of $X_1, \ldots, X_n$ is distributed to $J$
machines. For simplicity, we assume that the data is randomly
partitioned with equal size so that the sample size on each machine is $m = n/J$. 
Second, individual analysis is performed to the subset data on
each machine, usually in a paralleled fashion. The last step is to combine the estimators from all $J$
machines. The computational cost of a divide and conquer algorithm
is reduced tremendously thanks to the paralleled analyses on 
much smaller data sets. The reduction in time could be
significant if the complexity of the statistical analysis is of first
or higher order of sample size.
In the Bayesian framework, different approaches
arise in this context for various purposes. To name a few,
\cite{scott2016bayes} came up with a simple procedure in terms of
both assigning prior and combining posterior
samples. \cite{srivastava2015scalable} unified all posterior
distributions on each subset, leading to an overall posterior
distribution that maintains the same concentration rate as if the
whole data has been treated together. \cite{sabnis2016divide} made a
first step in subsetting the variables in a Bayesian factor model. \cite{guhaniyogi2017divide} proposed distributed kriging for Gaussian process in spatial data.

It is well acknowledged in both frequentist and Bayesian perspectives
that some debiasing or overfitting procedure needs to be done when
analyzing the subset data, in order to obtain a combined estimator
that achieves the same accuracy as that of the original estimator when
the data is analyzed as a whole.
Methods are distinguished by how the individual analysis is
appropriately adjusted and the way that the estimators in different
machines are combined. Here our attention is given to several recent
Bayesian approaches. Within the context of signal-in-white-noise
model, \cite{szabo2017asymptotic} pointed out the necessity of
carefully choosing among several strategies for successfully achieving
the optimal convergence rate and posterior coverage probability.
\cite{scott2016bayes} and
\cite{neiswanger2014asymptotically} proposed a general
framework for modifying the prior in the divide and conquer context
and applied the method in various setups.  Other approaches
\citep{srivastava2015scalable,xue2019double} concern modifying the
likelihood in evaluating the posterior distribution together with
combination techniques that find the ``center" of the posterior
distributions on each subset. However, these methods require caution
to use when the parameter is of large or infinite
dimension, with the combination strategy too simple to be justified or too
complicated to compute. In addition, these combination techniques are
also impossible to apply when the goal is to estimate a
density on a non-Euclidean space, e.g. space of densities. A comparison with this estimator, named WASP, is of merit in the low dimension case.
 

Of the existing methods using a Bayesian procedure along with a data-splitting technique, none of them estimate an infinite dimensional
parameter nor do they deal with prior distributions over an infinite
dimensional space. Problems of this kind arise naturally in
nonparametric density estimation and become attractive in high
dimensional and nonlinear models, see Section \ref{other_app} for an
incomplete list of references. In such problems, the most popular choice of
prior one can expect is the Dirichlet process mixture of standard
densities, like Normal, Laplace, Gamma, etc. Following the
success of \cite{scott2016bayes} in examples with a finite dimensional parameter or a simple
conjugate prior, we are motivated to
take a step into more complicated scenarios, for instance, when the prior
belongs to the specific type as above.  

In the divide and conquer framework, we propose a general methodology
for assigning such priors with a focus on density estimation. The
proposed prior generalizes the idea when the parameter is finite
dimensional, that is, the prior is adjusted with the purpose of
debiasing the subset density estimators by sacrificing the uncertainty
therein. We provide the solutions to adjust priors having a Dirichlet
process component and more. We use simple averaging for combining the
individual estimators which facilitates the computation and meanwhile reduces
excessive uncertainty without introducing bias. The ultimate density estimator after combining the individual estimators is constructed under much less computational and memory burden while still achieving the optimal rate. We show in simulations and real data that the proposed procedure is also applicable to contexts beyond density estimation. Included in this paper is a confirmed success in a density deconvolution problem. By its design, the method fits easily into other contexts as long as the prior itself is composed of a(n) finite/infinite mixture of standard probability distributions and others.

The following sections are organized in this way. Section \ref{sec2} specifies two example models regarding Bayesian density estimation and the proposed priors in the divide and conquer context.  In Section \ref{sec3} we show that the proposed density estimator can achieve the optimal rate when the number of subsets is growing no faster than $\log$ rate of sample size. Two simulations are conducted in Section \ref{sec4}, illustrating the competency of our method with WASP in density estimation and the capability of an extension to density deconvolution which is motivated by a real application. The proposed method is implemented on a data set in genome-wide association study (GWAS) and results are presented in Section \ref{sec6}. The paper ends with discussions in Section \ref{sec7}.

{\bf Notations}. $\mbox{IW}(\nu, \sbf{\Psi})$ denotes an inverse wishart distribution with degrees of freedom $\nu$ and scale matrix $\sbf{\Psi}$.  $\mbox{Ga}(a, b)/\mbox{IG}(a, b)$ stands for a(n) gamma/inverse gamma distribution with shape $a$ and rate/scale $b$. $\text{Dir}(\sbf{\alpha})$ denotes a Dirichlet distribution of order $K$ with parameter $\sbf{\alpha} = \{\alpha_1, \ldots, \alpha_K\}$. $\mbox{N}_p(\sbf{\mu},\sbf{\Sigma})$ stands for a multivariate normal distribution of dimension $p$ with mean $\sbf{\mu}$ and covariance matrix $\sbf{\Sigma}$ and $\mbox{N}(\mu,\sigma^2)$ in the case $p=1$. $\mbox{DP}(M, G)$ denotes a Dirichlet process with concentration parameter $M$ and base probability measure $G$. $\text{Unif}(\theta_1, \theta_2)$ is a uniform distribution supported on the interval $[\theta_1,\theta_2]$.  All of the distributions above can be easily switched to a density by adding a dot argument as the first argument, e.g. $\mbox{N}_p(\cdot; \sbf{\mu},\sbf{\Sigma})$ denotes the density of $\mbox{N}_p(\sbf{\mu},\sbf{\Sigma})$. By convention, $\phi_\sigma$ refers to the density function of a univariate Normal distribution with mean $0$ and standard deviation $\sigma$.  The expression $a_n \asymp b_n$ states that $a_n$ and $b_n$ are of the same rate asymptotically.

\section{Model Specification}\label{sec2}

\subsection{Background}\label{sec2.1}
Suppose $X_1,\ldots,X_n \in \mathbb{R}^p$ is an independent and identically distributed sample from an unknown density $f_0(\cdot), \mathbb{R}^p \rightarrow [0,\infty)$. We are going to illustrate our idea under both parametric and nonparametric model for $f_0(\cdot)$, both characterized by a mixtures of normal distributions. The key idea can be easily extended to mixtures of distributions other than normal.

We first consider $f_0(\cdot)$ is a finite dimensional mixture of normal distributions, 
\begin{equation}
  \label{DDmixmvn}
  f_0(x) = \sum_{k=1}^K \pi_k \mbox{N}_p(x; \sbf{\mu}_k, \sbf{\Sigma}_k).
\end{equation}
The mixing probabilities $\sbf{\pi}=(\pi_1,\ldots,\pi_K)$ lie in a $(L-1)$-simplex. The $k$th component normal distribution has mean and covariance matrix $\sbf{\mu}_k$ and $\sbf{\Sigma}_k$. Together, $\sbf{\pi}$, $\{\sbf{\mu}_k\}_{k=1}^K$ and $\{\sbf{\Sigma}_k\}_{k=1}^K$ form the unknown parameters in the data generating model.

For ease of computation, a conjugate prior corresponding to
(\ref{DDmixmvn}) can be imposed \citep{srivastava2015scalable}:
\begin{eqnarray}
  \label{DDmixmvn_prior}
  \sbf{\pi} \sim \text{Dir}(\alpha_1,\ldots,\alpha_K), \quad 
  \sbf{\mu}_k \mid \sbf{\Sigma}_k \sim \mbox{N}_p(\sbf{0}, l \sbf{\Sigma}_k), \quad   
  \sbf{\Sigma}_k \sim \mbox{IW}(\nu, \sbf{S}).
\end{eqnarray}

On the other hand, if the form of the true density $f_0(\cdot)$ is unknown, in which case a nonparametric counterpart to the finite dimensional model (\ref{DDmixmvn}) and (\ref{DDmixmvn_prior}) is a popular substitute.  We are going to present the univariate case, a straightforward extension of the current algorithm and theory to a multivariate case exists; see Remark \ref{thm_mv} for a brief discussion about the theory about the multivariate case. Specifically, the nonparametric model is
\begin{eqnarray}
  \label{DPmixmvn}
   f(x) \sim \int \phi_\sigma(x-\mu) P(d\mu), \quad  \sigma \sim
  \Pi_\sigma(\cdot),  \quad  P \sim \mbox{DP}(M,G).  
\end{eqnarray}
The model (\ref{DPmixmvn}) corresponds to the so-called Dirichlet
process (location) mixtures of Normal (DPMN) prior, algorithms of which have been studied
previously \citep{rasmussen2000infinite, blei2006variational}. The
asymptotics about DPMN have been investigated in
\cite{ghosal2007posterior} for the univariate case and \cite{shen2013adaptive}
for the multivariate case.

Given the proven performance of these priors in producing a good
density estimator while running on the complete data, the following
sections will provide guidance on imposing priors when we work on
small chunks of data spread across various machines. Before
illustrating our approach regarding the density estimation problems
above, we first present a general way which has been explored for
models with finite dimensional parameters \citep{scott2016bayes}. Denote $\theta$ be the parameter of interest, $L(\theta; \sbf{x})$ be the likelihood function based on data $\sbf{x}$, $\Pi(\theta)$ be the prior on $\theta$. In the distributed setting with $J$ chunks, the likelihood function can be decomposed into $J$ components, $L(\theta; \sbf{x}) = \prod_{j=1}^J L(\theta; \sbf{x}^{(j)})$, where $\sbf{x}^{(j)}$ is the data belonging to the $j$th chunk.  The posterior distribution of $\theta$ takes the form $\Pi(\theta\mid\sbf{x}) = \prod_{j=1}^{J} \{L(\theta; \sbf{x}^{(j)}) \Pi(\theta)^{1/J}\}$ where the likelihood function and the prior are factorized similarly. This general idea paves the way of seeking for an appropriate prior on each chunk of data by assigning $\Pi(\theta)^{1/J}$.

Difficulties exist on how to justify the above idea for all cases of
$\Pi(\theta)$ where $\theta$ could be infinite dimension.  The focus
of this paper is to address this issue for a family of models
including but not limited to (\ref{DDmixmvn}) and (\ref{DPmixmvn}).
It is seemingly hard to handle priors with its support on a probability space with the existing literature because these priors involve a a ``distribution on distribution" component corresponding to the Dirichlet distribution or the Dirichlet process prior and some independent prior distributions on the rest of parameters in the component densities. In what follows we describe the modification to these priors in the divide and conquer context. The prior on the parameters in component distributions which usually takes a conjugate form against the likelihood will be imposed as the same type. On the other hand, we propose to make a simple adjustment on the parameters of the Dirichlet prior from the property of Dirichlet distribution or process.

\subsection{Finite mixtures of Normal prior}\label{fmnp}
As introduced in Section \ref{sec2.1}, the prior for a finite mixture of normal
model takes the form \eqref{DDmixmvn_prior}. We start with a basic
property of the Dirichlet distribution.

\begin{proposition}\label{DD_moment}{\rm
  Let $\sbf{\pi} \sim \text{Dir}(\sbf{\alpha})$, where
  $\sbf{\alpha} = (\alpha_1,\ldots,\alpha_K)$. Denote $\alpha_s =
  \sum_{k=1}^K \alpha_k$. Then
  $E(\pi_k)=\alpha_k/\alpha_s$,
  $\text{var}(\pi_k)=\alpha_k(\alpha_s-\alpha_k)/\alpha_s^2
  (\alpha_s+1)$, $\text{cov}(\pi_j,\pi_k)=-\alpha_j \alpha_k/\alpha_s^2
  (\alpha_s+1)$. 
 }
\end{proposition}

The Dirichlet distribution prior on $\sbf{\pi}$,
$\text{Dir}(\alpha_1,\ldots,\alpha_K)$, is proportional to
$\prod_{k=1}^K \pi_k^{\alpha_k - 1}$. Proposition \ref{DD_moment} (a
proof can be found in Chapter 27 of \cite{balakrishnan2004primer})
states that scaling $\sbf{\alpha}$ by a factor $\gamma$ scales up the variance by a factor $(1 + \alpha_s)/(1 +
\gamma \alpha_s)$ while keeping the mean unchanged component-wise. This property is essential and provides
directions to adjust the prior on $\sbf{\pi}$. We propose to scale the
parameters of the Dirichlet distribution prior on $\sbf{\pi}$ by $1/J$ leading to
$\text{Dir}(\alpha_1/J,\ldots,\alpha_K/J)$.  Indeed, we can make an
assertion that this simple adjustment sacrifices
  uncertainty in exchange for debiasing regarding $\sbf{\pi}$
according to Proposition \ref{DD_moment}.

We take the same strategy as Concensus Monte Carlo
  \citep{scott2016bayes} for adjusting the
priors on the remaining parameters. Raising a power $1/J$ to the normal prior on $\sbf{\mu}_k \mid \sbf{\Sigma}_k$ leads to
$\mbox{N}_p(0, l \sbf{\Sigma}_k J)$. It can be shown easily that the effect of a power
$1/J$ to the inverse wishart prior on $\sbf{\Sigma}_k$, $\mbox{IW}(\nu, \sbf{S})$, is chracterized by an inverse wishart type $\mbox{IW}(\nu/J - (p+1)(J-1)/J, \sbf{S}/J)$. The risk
of its first parameter being possibly negative will be regulated by the likelihood when
constructing the (conditional) posterior distribution of $\sbf{\Sigma}_k$. A Gibbs sampler corresponding to the
modified prior above is provided in Section S.1.1 in the
Supplementary material.

\subsection{Dirichlet process mixtures of Normal prior}\label{dpmnp}
Although a Dirichlet process has a remarkable stick-breaking representation \citep{sethuraman1994constructive}, unfortunately, it does not have a probability density as a Dirichlet distribution does. However, we can extend the idea in Section \ref{fmnp} to Dirichlet process since marginally a Dirichlet process follows a Dirichlet distribution. That is, if $P \sim \mbox{DP}(M,G)$, for any measurable finite partition $\{A_1, \ldots, A_k\}$ of the support of the base measure $G$, $(P(A_1), \ldots, P(A_k)) \sim \text{Dir}(G(A_1)M,\ldots,G(A_k)M)$.

Our idea is to modify the parameters associated with a Dirichlet process such that the relationships between its subsequent marginal distributions and those under the Dirichlet process with the original parameters are maintained to be the same as that in Section \ref{fmnp}.  This can achieved by adjusting the prior on $P$ as $\mbox{DP}(M/J,G)$ under which $(P(A_1), \ldots, P(A_k)) \sim \text{Dir}(G(A_1)M/J,\ldots,G(A_k)M/J)$ for the above partition. Hence regarding the nonparametric model \eqref{DPmixmvn}, the following prior is suggested:
\begin{eqnarray}
  \label{DPmixmvn_pprior}
   f(x) \sim \int \phi_\sigma(x-\mu) P(d\mu), \quad \sigma \propto                                  
     \Pi_\sigma^{1/J}(\cdot), \quad  P \sim \mbox{DP}(M/J,G).  
\end{eqnarray}
The general form of the prior on $\sigma$, $\Pi_\sigma^{1/J}$, can be simplified if $\Pi_\sigma$ takes a parametric form. In the case when a conjugate prior for $\sigma$, $\Pi_\sigma \sim \mbox{IG}(a, b)$, is adopted, it becomes $\Pi_\sigma^{1/J} \sim \mbox{IG}(a/J-(J-1)/J, b/J)$. Indeed it is just an inverse gamma type since the first argument is negative when $a < (J - 1)$, but similar arguments about the prior on $\sbf{\Sigma}_k$ in Section \ref{fmnp} apply here. A Gibbs sampler corresponding to \eqref{DPmixmvn_pprior} is provided in Section S.1.2 in the Supplementary material.

\subsection{A combined density estimator}\label{concept_in_thm}

Let $\sbf{X}_j$ be the subset of data distributed to the $j$th machine, $j = 1, \ldots, J$. Denote the posterior probability under the model in Section \ref{fmnp} or \ref{dpmnp} in accordance to the $j$th subset as $\Pi_m(\cdot \mid \sbf{X}_j)$, where subscript $m$ indicates the distributed sample size. Our procedure proceeds as follows. For each subset $j$, we could obtain an estimator $f^{[j]}(\cdot) \in \mathcal{P}$ by taking a random sample from the posterior distribution $\Pi_m(\cdot \mid \sbf{X}_j)$.  The ultimate estimator is then formed by a simple average over all subset samples, $\widebar{f}(\cdot) = (1/J) \sum_{j=1}^J f^{[j]}(\cdot)$.

Let $\widebar{\Pi}_n(\cdot \mid \sbf{X}_n)$ denote the distribution of
$\widebar{f}(\cdot)$. Then it is easy to show that
$\widebar{\Pi}_n(\cdot \mid \sbf{X}_n)$ is a convex convolution of all
subset posterior densities, $\Pi_m(\cdot \mid \sbf{X}_j)$, $j = 1,
\ldots, J$. Although we will provide asymptotics of $\widebar{\Pi}_n(\cdot \mid
\sbf{X}_n)$ in Section \ref{main}, we discuss the appropriateness
of proposing $\widebar{f}(\cdot)$, or equivalently $\widebar{\Pi}_n(\cdot \mid \sbf{X}_n)$. The
major consequence of us modifying the prior is that the center of $\Pi_m(\cdot \mid
\sbf{X}_j)$ is pulled towards the targeted posterior distribution
while admitting larger variability. We construct the aforementioned combined density
/posterior distribution as it keeps the center unchanged while reducing the
variability in subset posterior distributions.

\subsection{Other applications}\label{other_app}
Our divide and conquer algorithm goes beyond density estimation
problem as long as the prior of the relevant model consists of 
a Dirichlet distribution/process component, which is often seen along with
models characterized by a(n) finite/infinite mixture of standard probability
distributions. The popularity of such prior has risen in recent years
with appearances in high dimensional normal means problem
\citep{bhattacharya2015dirichlet}, multivariate categorical data
with dependency \citep{dunson2009nonparametric}, and nonlinear
regression models \citep{de2010adaptive, naulet2018some}, just to name
a few.  

\section{Theoretical results}\label{sec3}

\subsection{Preliminary definitions}

The set of density functions is $\mathcal{P} = \{f(\cdot):
f(\cdot) \geq 0, \int
f(x) dx = 1\}$. We consider the metric on $\mathcal{P}$ to be the Hellinger distance
$h(\cdot, \cdot)$. For any $p, q \in \mathcal{P}$, $h(p, q) = \{\int (p^{1/2}(x) - q^{1/2}(x))^2 dx\}^{1/2}$. The Wasserstein space
of order $2$ is defined as $\mathscr{P}_2 = \{\mu \in
\mbox{probability measure on } \mathcal{P}: \int_{\mathcal{P}} h^2(f, f_0) d\mu(f) <
\infty\}$. For any $\mu, \nu \in \mathscr{P}_2$, $\Xi(\mu, \nu)$ is a
set of all probability measures on $\mathcal{P} \times \mathcal{P}$
whose marginal measures are $\mu$ and $\nu$. The Wasserstein distance
of order $2$ on $\mathscr{P}_2$ is defined as $W_2^2(\mu, \nu) = \inf_{\xi \in \Xi(\mu, \nu)} \int_{\mathcal{P} \times \mathcal{P}} h^2(p, q) d\xi(p, q)$.

In particular, if one of the probability measures is concentrated on a
fixed element in $f_0 \in \mathcal{P}$, e.g., $\nu = \delta_{f_0}$,
the Wasserstein metric becomes $W_2^2(\mu, \nu) = \int_{\mathcal{P}}
h^2(f, f_0) d\mu(f)$. Thus in the context of this paper, the Wasserstein distance between
any posterior distribution $\Pi_n(\cdot \mid \sbf{X}_n)$ on
$\mathcal{P}$ and $\delta_{f_0}$ is  
\begin{eqnarray}
  \label{wass_dist}
W_2^2(\Pi_n,\delta_{f_0}) = \int_{\mathcal{P}} h^2(f, f_0) d \Pi_n(f
  \mid \sbf{X}_n).
\end{eqnarray}
Hence
\begin{align*}
  W_2^2(\Pi_n,\delta_{f_0}) \leq \epsilon_n^2 + 2 \Pi_n(\{f \in
  \mathcal{P}: h(f, f_0) \geq \epsilon_n\} \mid \sbf{X}_n).
\end{align*}
The above inequality is due to the fact that the Hellinger distance is less
or equal to $2$. Hence we conclude that a typical posterior
contraction rate result, $\Pi_n(\{f \in \mathcal{P}: h(f, f_0) \geq
\epsilon_n\} \mid \sbf{X}_n) \rightarrow 0$,
  is sufficient to prove a convergence of $\Pi_n(\cdot \mid \sbf{X}_n)$ to
  $\delta_{f_0}$ with rate $\epsilon_n$ in Wasserstein distance. In Section \ref{main} we will present our
  theoretical results in terms of the latter.   

  \begin{lemma}\label{conv_h}{\rm
  The Hellinger distance is a biconvex functional on $\mathcal{P} \times
  \mathcal{P}$, that is, for any
  $g_1, g_2, f \in \mathcal{P}$, and $\omega_1 \geq 0, \omega_2 \geq 0$ with
  $\omega_1 + \omega_2 = 1$, $h(\omega_1 g_1 + \omega_2 g_2, f)
  \leq \omega_1 h(g_1, f) + \omega_2 h(g_2, f)$. Similarly, $h(f, \omega_1 g_1 + \omega_2 g_2)
  \leq \omega_1 h(f, g_1) + \omega_2 h(f, g_2)$.}
  \end{lemma}

\subsection{Main theorems}\label{main}

To see the asymptotic behavior of $\widebar{f}(\cdot)$, equivalently, we can study the underlying distribution which yields $\widebar{f}(\cdot)$, $\widebar{\Pi}_n(\cdot \mid \sbf{X}_n)$. By definition of $\widebar{f}(\cdot)$, $\widebar{\Pi}_n(\cdot \mid \sbf{X}_n)$ is a (convex) convolution of the subset posterior distributions $\Pi_m(\cdot \mid \sbf{X}_j), j = 1, \ldots, J$. Specifically, for any functional $L(\cdot)$ on $\mathcal{P}$,
\begin{eqnarray} \label{convex_conv} \int_{\mathcal{P}} L(f) d \widebar{\Pi}_n(f \mid \sbf{X}_n) = \int_{\mathcal{P} \times \cdots \times \mathcal{P}} L({\textstyle \sum\nolimits_{j=1}^J} f_j/J) d \Pi_m(f_1 \mid \sbf{X}_1) \cdots d\Pi_m(f_J \mid \sbf{X}_J).
\end{eqnarray} 
It is trivial to show that $\widebar{f}(\cdot)$ corresponds to a sample drawn from $\widebar{\Pi}_n(\cdot \mid \sbf{X}_n)$.

We are going to state the posterior contraction rate for
$\widebar{\Pi}_n(\cdot \mid \sbf{X}_n)$. For illustration purposes, we are
going to state the theory for the nonparametric density estimator under the Dirichlet process mixtures of normal prior \eqref{DPmixmvn_pprior}. 

Let $a_1, \ldots, a_5$, $b_1, \ldots, b_5$ and $C_1, \ldots, C_4$ be
positive constants. Denote $C^{\beta, L, \tau_0}$
as the locally $\beta$-H\"older function class with functions that
have finite partial derivatives $f^{(k)}(\cdot)$ up to order $k \leq \lfloor \beta \rfloor$
such that for all $k \leq \lfloor \beta \rfloor$, $
 |f^{(k)}(x + y) - f^{(k)}(x)| \leq L(x) \exp(\tau_0|y|^2) |y|^{\beta
   - \lfloor \beta \rfloor}$.

The key assumptions are \\
(C1) $1 - G([-x,x]) \leq b_1 \exp(-C_1 x^{a_1})$ for sufficiently large $x > 0$.\\
(C2) $\Pi_\sigma\{(0, x)\} \leq b_2 \exp(- C_2 x^{-a_2})$, for
sufficiently small $x > 0$. $\Pi_\sigma\{(x, \infty)\} \leq b_3
x^{-a_3}$, for sufficiently large $x > 0$. For any $s, t > 0$
$\Pi_\sigma\{(s^{-1}(1 + t)^{-1/2}, s^{-1})\} \geq b_4 s^{-a_4} t^{a_5} \exp(-C_3 s^{-1})$. \\
(C3) $f_0(\cdot) \in C^{\beta, L, \tau_0}$. $P_0(|D^kf_0|/f_0)^{(2
  \beta + \epsilon)/k} < \infty$ for all integer $k \leq \lfloor \beta
\rfloor$. $P_0(L/f_0)^{(2 \beta + \epsilon)/ \beta} < \infty$. $f_0(x)
\leq b_5 \exp(-C_4|x|^\tau)$ for sufficiently large $|x|$ and some
$\tau > 0$. \\
(C4) $J \asymp \log n$, equivalently, $m \asymp n \log^{-1} n$.  

The first three conditions are the same as those in
\cite{shen2013adaptive} in the univariate case. The last condition is on the growing rate of subsets in our divide and conquer setup.

\begin{thm}\label{conv_rate_sub}{\rm
If assumptions (C1)\textemdash(C4) are satisfied, for any
$j=1,\ldots,J$ and any $t > 3/2 + (1/\tau)/(2 + 1/\beta)$, the
posterior distribution $\Pi_m(\cdot \mid \sbf{X}_j)$ converges to $f_0(\cdot)$ in Wasserstein distance with contraction rate $\epsilon_m = m^{-\beta/(2\beta + 1)} (\log m)^t$, that is,
\begin{eqnarray*}
  W_2(\Pi_m(\cdot\mid\sbf{X}_j), \delta_{f_0}) \lesssim \epsilon_m, \quad j=1,\ldots,J.
\end{eqnarray*}
Here $\delta_{f_0}$ is the dirac measure on $\mathcal{P}$ concentrating at $f_0(\cdot)$.
}
\end{thm}

The next lemma states how the convergence under the Wasserstein metric
for each subset posterior distribution $\Pi_m(\cdot \mid \sbf{X}_j)$
controls that of the posterior distribution $\widebar{\Pi}_n(\cdot
\mid \sbf{X}_n)$.
\begin{lemma}\label{conv_w}
  $W_2(\widebar{\Pi}_n(\cdot \mid \sbf{X}_n), \delta_{f_0}) \leq
  J^{-1} {\textstyle \sum_{j=1}^J} W_2(\Pi_m(\cdot \mid \sbf{X}_j), \delta_{f_0})$.
\end{lemma}
\begin{proof}
First, from the expressions
\eqref{convex_conv}, \eqref{wass_dist} and Lemma \ref{conv_h}, it can be seen that 
\begin{eqnarray*}
  W_2^2(\widebar{\Pi}_n(\cdot \mid \sbf{X}_n), \delta_{f_0}) &=& \int_{\mathcal{P} \times \cdots \times \mathcal{P}}
  h^2({\textstyle \sum_{j=1}^J}
  f_j/J, f_0) d \Pi_m(f_1 \mid \sbf{X}_1) \cdots d\Pi_m(f_J \mid \sbf{X}_J)\\
  &\leq& J^{-2} {\sum_{j=1}^J} \int_{\mathcal{P}} h^2(f, f_0) d\Pi_m(f \mid \sbf{X}_j) \\
  &+& J^{-2} {\sum_{j \neq k}}
      \int_{\mathcal{P} \times \mathcal{P}}h(f_j, f_0) h(f_k, f_0)
      d\Pi_m(f_j \mid \sbf{X}_j) d\Pi_m(f_k \mid \sbf{X}_k).
\end{eqnarray*}  
By Cauchy-Schwartz inequality, all terms in the second summation on the right hand side
is bounded by $\{\int_{\mathcal{P}}h^2(f, f_0) d\Pi_m(f \mid
\sbf{X}_j)\int_{\mathcal{P}}h^2(f, f_0) d\Pi_m(f \mid
\sbf{X}_k)\}^{1/2}$. Applying expression \eqref{wass_dist} again to all
terms in the previous display yields
\begin{eqnarray*}
W_2^2(\widebar{\Pi}_n, \delta_{f_0}) &\leq& J^{-2} \bigg\{ {\sum_{j=1}^J}
  W_2^2(\Pi_m(\cdot \mid \sbf{X}_j), \delta_{f_0}) + { \sum_{j
  \neq k}}  W_2(\Pi_m(\cdot \mid \sbf{X}_j),\delta_{f_0}) W_2(\Pi_m(\cdot \mid 
  \sbf{X}_k),\delta_{f_0}) \bigg\}\\
  &=& \bigg\{J^{-1} {\sum_{j=1}^J} W_2(\Pi_m(\cdot \mid \sbf{X}_j), \delta_{f_0})\bigg\}^2.
\end{eqnarray*}
\end{proof}

\begin{thm}\label{main_thm}{\rm
If assumptions (C1)\textemdash(C4) are satisfied, the posterior
distribution $\widebar{\Pi}_n(\cdot \mid \sbf{X}_n)$ converges to $f_0(\cdot)$ in Wasserstein distance with contraction rate $\epsilon_n = n^{-\beta/(2\beta + 1)} (\log n)^{u}, u > 3/2 + (1/\tau)/(2 + 1/\beta) + \beta/(2\beta + 1)$, that is,
\begin{eqnarray*}
  W_2(\widebar{\Pi}_n(\cdot\mid\sbf{X}_n), \delta_{f_0}) \lesssim \epsilon_n.
\end{eqnarray*}
}  
\end{thm}

\begin{proof}
  The desired inequality is implied by Theorem \ref{conv_rate_sub} and Lemma~\ref{conv_w}. 
\end{proof}

\begin{remark}\label{thm_mv}{\rm
In Theorem \ref{main_thm}, we state the one dimensional case for simplicity. We make a brief comment without details that a multivariate version of Theorem \ref{main_thm} exists 
provided a multivariate correspondence to prior
\eqref{DPmixmvn_pprior} is imposed. In the standard Bayesian density
estimation setting, the contraction rate for estimating a multivariate
density has been studied previously in \cite{shen2013adaptive}. When
switched to the multivariate case with dimension $d$, the correspondence to 
$\Pi_\sigma$ would be a prior on the covariance matrix $\Sigma \in
\mathbb{R}^{d \times d}$, on which
\cite{shen2013adaptive} imposed conditions regarding the
concentration of eigenvalues of the matrix. Meanwhile, the base
measure $G(\cdot)$ on $\mathbb{R}^d$ is required to satisfy a straightforward
multi-dimensional extension of (C1). 
Under these conditions, the logic flow in deriving Theorem~\ref{main_thm} guarantees that our divide and conquer density estimator can achieve the multivariate contraction rate as if the complete data has been used. Details will be omitted.
 }
\end{remark}

\section{Simulation}\label{sec4}

\subsection{Overview}\label{sec5.1}

As briefed in Section \ref{fmnp}--\ref{concept_in_thm}, we come up
with a proper way to rebuild a simple density estimator for some existing 
parametric and nonparametric methods in the divide and
conquer context. 

The finite mixtures of normals is one example considered in \cite{srivastava2015scalable} regarding model \eqref{DDmixmvn} and prior \eqref{DDmixmvn_prior}. For comparison, our first simulation will imitate Section 4.2 of \cite{srivastava2015scalable} in the choices of $p = 2$, $K = 2$, $\alpha_1 = \alpha_2 = 1/2$, $l = 100$, $\nu = 2$ and $\sbf{S} = 4 \sbf{I}_2$. The true parameter values are set as $\sbf{\pi} = (0.3, 0.7)$, $\sbf{\mu}_1 = (1, 2)^T$, $\sbf{\mu}_2 = (7, 8)^T$, $\sbf{\Sigma}_1 = \sbf{\Sigma}_2 = \{\Sigma_{ij}\}_{1 \leq i,j \leq 2}$, where $\Sigma_{12} = 0.5$, $\Sigma_{11} = 1$ and $\Sigma_{22} = 2$.  As mentioned in Section \ref{fmnp}, we propose the following adjusted prior on each machine:
\begin{eqnarray}
  \label{ad_DDmixmvn}
  \sbf{\pi} \sim \text{Dir}(1/(2J), 1/(2J)), \quad 
  \sbf{\mu}_k \mid \sbf{\Sigma}_k \sim \mbox{N}_2(\sbf{0}, 100 \sbf{\Sigma}_k J), \quad   
  \sbf{\Sigma}_k \sim \mbox{IW}(5/J-3, (4/J) \sbf{I}_2).
\end{eqnarray}

\cite{srivastava2015scalable} mainly illustrated the ability of their method, WASP, on estimating some nonlinear functions of the parameters, say $g(\sbf{\pi}, \sbf{\mu}_1, \sbf{\mu}_2, \sbf{\Sigma}_1, \sbf{\Sigma}_2)$. They compared the accuracy of WASP with other methods, among which the superiors were Consensus Monte Carlo and WASP. Given their close competition in estimating the functions of the parameters, it is worthwhile to examine whether the performance is consistent in estimating the original model parameters. With the acknowledgment of the possible similarity of our method and Consensus Monte Carlo in finite dimensional models, we implement our algorithm using (\ref{ad_DDmixmvn}).  A disclaimer is that we find there are differences with the implementation of CMC in \cite{srivastava2015scalable} (specifically in updating $\sbf{\pi}$ and $\Sigma_k$). WASP estimator is obtained using online code of \cite{srivastava2015scalable}.

The second simulation we conduct is in a more complicated context, as designed for density deconvolution with shape constraints. We generate observed data $W_i$ through a classical measurement error model, specifically, an independent sample of $W = X + U$, where the distribution of $U$ is known and possibly heteroscedastic. The density of the true variable $X$, $f(\cdot)$, is of interest and in some application context (see Section \ref{sec6} for one such application) should have a symmetric and unimodal shape.  To ensure the shape of $f(\cdot)$, a multi-layer mixture prior is adopted. According to \cite{feller1971introduction}, any unimodal and symmetric density (with finite first derivative) $f(\cdot)$ can be represented by a mixture of uniform distributions. We focus on the scenario that the mixing distribution has a density $g(\cdot)$. The density $g(\cdot)$ is then built upon a Dirichlet process mixture of gamma distributions. Using latent variables, we can write out the hierarchical model
\begin{eqnarray} \label{ddsc}
  \notag&W_i \mid X_i \sim \mbox{N}(X_i, \sigma_i^2);  \quad X_i \mid \theta_i \sim \text{Unif}(-\theta_i, \theta_i);  \quad \theta_i \mid z, \mu \sim \text{Ga}(z, z/\mu); \\
  &\mu \mid P_\mu \sim P_\mu; \quad P_\mu \mid m, D \sim \mbox{DP}(m,D); \quad z \sim \Pi_z(\cdot),
\end{eqnarray}  
where the gamma distributions are reparameterized by shape $z$ and
mean $\mu$. The mixing is imposed on $\mu$ leading to a Dirichlet
process location mixture of gamma distributions. The proposed prior 
works well in simulations and real data in a recent work under review
in a peer-reviewed journal. 

The second simulation setup involves a hierarchical prior through the introduced latent variables. We propose to
impose the fraction $1/J$ on the bottom layer of the prior, namely, 
\begin{eqnarray}
  \label{ddsc_dc}
  \notag&W_i \mid X_i \sim \mbox{N}(X_i, \sigma_i^2);  \quad X_i \mid \theta_i \sim \text{Unif}(-\theta_i, \theta_i);  \quad \theta_i \mid z, \mu \sim \text{Ga}(z, z/\mu); \\
  &\mu \mid P_\mu \sim P_\mu; \quad P_\mu \mid m, D \sim \mbox{DP}(m/J,D); \quad z \sim \Pi_z^{1/J}(\cdot),  
\end{eqnarray}
To the best of our knowledge, hierarchical priors have not
been investigated in the divide and conquer context. Thus it highlights the
capability of applying our method in a broad range. A Gibbs sampler
corresponding to \eqref{ddsc_dc} is provided in Section S.1.3
in the Supplementary material.

The true density is composed of a normal component, $\mbox{N}(0,0.2^2)$,
which targets a sharp peak at zero and a t distribution with
degrees of freedom $5$ which generates the large values of $X$. The
two components are assigned with probabilities $0.8$ and $0.2$ respectively
so that the resulting density has a sharp peak around zero and a
small portion on the large values. The choice of error variances is $\sigma_i^2 = (0.75 + X_i/4)^2$ under which the variance of error depends on $X$ and
the {\it expected} variance of error is more  
than the variance of $X$. 

We adopt posterior mean estimators unless specified otherwise. For WASP, we originally get an overall posterior distribution from the author's code, from which posterior samples are generated, this in turn leads to the posterior mean estimator. Our estimator is obtained first by taking the average of posterior samples of densities across the MCMC steps on each individual machine. Then our divide and conquer estimator is calculated by averaging over the estimators across the selected machines. In addition, an estimator based on the original analysis using the complete data is implemented for validating the above estimators.

\subsection{Simulation Results}\label{sec5.2}

For the first simulation, the sample size of the complete
data is $n = 10,000$ and the number of MCMC steps is $10,000$ with
burn-in steps $5,000$ and thinning every $5$th iterations. In
addition, $J = 10$ machines are chosen for splitting the data. The
simulation is repeated $10$ times. All these setups agree with 
\cite{srivastava2015scalable} for comparison with the WASP estimator. 

Table \ref{tab1} summarizes and compares the accuracy of parameter estimation for $\sbf{\mu}_1$ and $\sbf{\mu}_2$. Table \ref{tab2} contains the counterparts for $\sbf{\Sigma}_1$ and $\sbf{\Sigma}_2$.
For the covariance matrices estimation, we are limited to present the performance of our method and the method that uses the complete data. The WASP estimators for $\sbf{\Sigma}_1$ and $\sbf{\Sigma}_2$ are missing due to incapability of obtaining posterior samples from the online codes provided in \cite{srivastava2015scalable}.

\begin{table}
\centering
\begin{tabular}{ccccccc}
\hline\hline
  Parameter &\multicolumn{3}{c}{$\sbf{\mu}_1$}
  &\multicolumn{3}{c}{$\sbf{\mu}_2$}\\
Estimator & full &WASP &fPrior    & full &WASP &fPrior \\
\hline
bias ($\times 10^{-3}$)  & (-2, -2) & (-2, -2) & (-2, -2) & (1, 3) 
& (0, 3) & (1, 3)  \\
se ($\times 10^{-3}$)  & (4.7, 4.5) & (4.3, 4.6) & (4.7, 4.5) 
& (5.7, 7.5) & (5.8, 7.6) & (5.7, 7.5) \\
\hline\hline
\end{tabular}
\caption{The bias and standard error in estimating $\sbf{\mu}_1$ and $\sbf{\mu}_2$
  using complete data (full), WASP and our method (fPrior). The
  reported values are magnified by a factor $10^3$ so that less
  decimals are displayed.}
\label{tab1}
\end{table}

\begin{table}
\centering
\begin{tabular}{ccccccc}
\hline\hline
  Parameter &\multicolumn{2}{c}{$\Sigma_{11}$}
  &\multicolumn{2}{c}{$\Sigma_{12}$}  &\multicolumn{2}{c}{$\Sigma_{22}$}\\
Estimator & full &fPrior & full &fPrior & full &fPrior \\
\hline
bias ($\times 10^{-3}$)  & (2, -2) & (2, 0) & (2, -4) & (2, -3) 
& (6, 0) & (8, 3)  \\
se ($\times 10^{-3}$) & (4.2, 13.3) & (5.1, 7.0) & (3.5, 10.7) 
& (4.7, 9.5) & (13.3, 7.6) & (10.6, 11.8)  \\
\hline\hline
\end{tabular}
\caption{The bias and standard error in estimating components of
  $\sbf{\Sigma}_1$, $\sbf{\Sigma}_2$ ($\sbf{\Sigma}_1 = \sbf{\Sigma}_2$)
  using complete data (full) and our method (fPrior). We report two values for each parameter
  corresponding to estimators from the two components. The
  reported values are magnified by a factor $10^3$ so that less
  decimals are displayed}
\label{tab2}
\end{table}

The second simulation is performed on a data set of size $15,000$ and repeated $50$ times. We use $5,000$ MCMC iterations with $1,000$ burn-ins. The density estimators are constructed in three contexts: the complete data with the original prior, divided data on $20$ machines with the original prior, divided data on $20$ machines with the proposed prior. For the latter two, the estimated density is obtained by further averaging the posterior densities over the $20$ machines. Figure \ref{Fig1} presents the estimated densities and the true density.

It is worthwhile to mention that the accuracy of our proposed prior is achieved with a significant computational gain. The algorithm runs on multiple computer nodes in a Linux OS cluster with single core assigned for individual analysis. In addition, less memory is assigned for analysis of split data than that used for analyzing the complete sample. Roughly, the overall time is about $1/5$ over that without data splitting.


 \begin{figure}[htbp]
  \centering
  \includegraphics[height=0.5\textwidth,
    width=0.8\textwidth, trim=0.5cm 1cm 1cm 2cm,
    clip=TRUE]{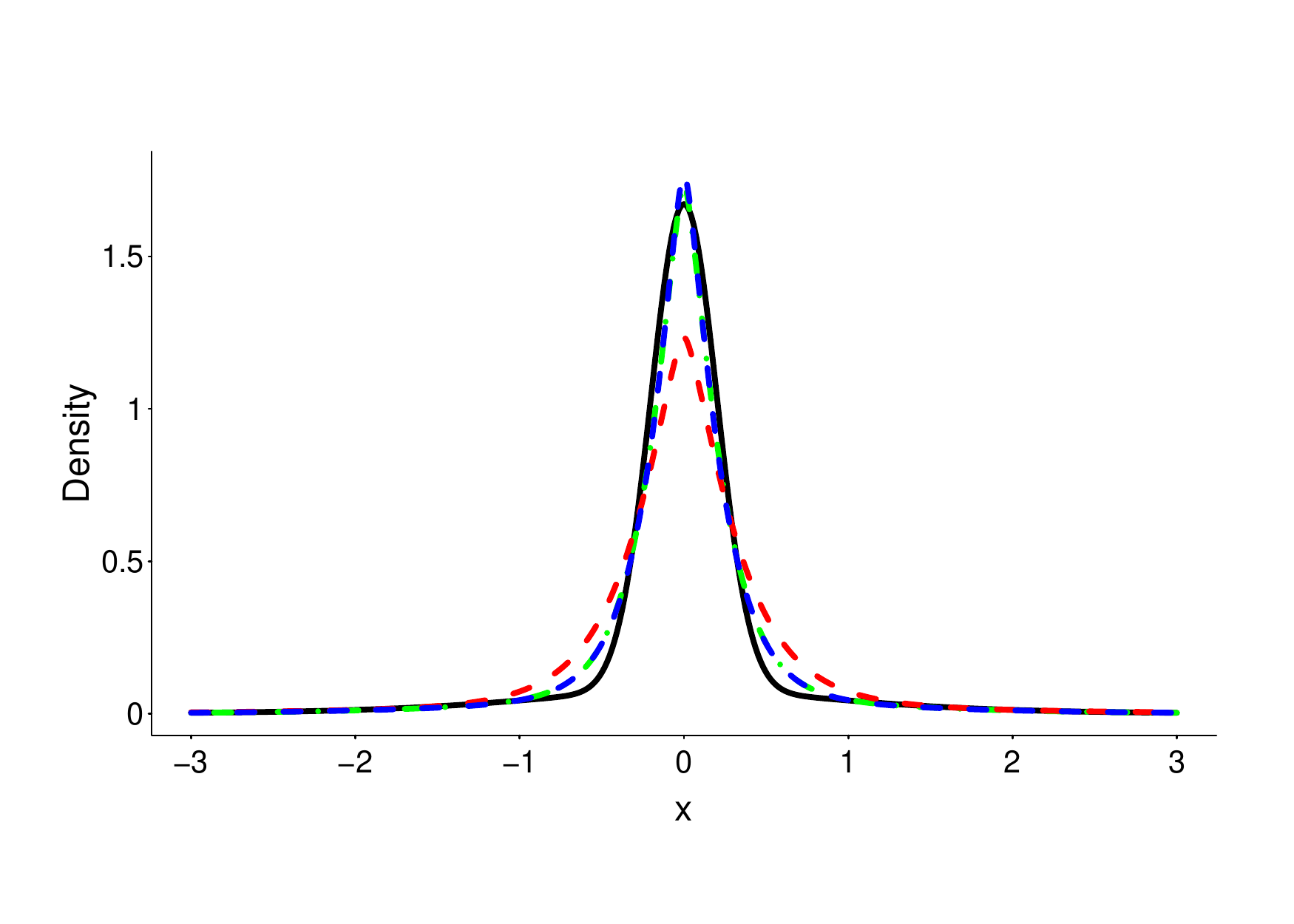}
    \caption{Deconvoluted density estimators using the complete data of size $15,000$ and using the divided data (of equal size) on $20$ machines. The truth (black solid line) is a mixture of $t$ with degrees of freedom $5$ with probability $0.2$ and $\mbox{N}(0,0.2^2)$ with probability $0.8$. The posterior mean of MCMC samples of the density using the complete data and the original prior (green dotted dash line) can be treated as the ``best" estimator. The averaged estimator using the original prior over $20$ machines (red dashed line) and our estimator (blue dashed line) are displayed for comparison.}
  \label{Fig1}
\end{figure}

\subsection{Conclusions}\label{sec5.3}
We implemented two distinct simulations in density estimation using finite mixture of normal distributions and density deconvolution with shape constraints. In the first simulation, the performance of our methods in estimating the mean parameter of normal distributions is very competitive with the established WASP while both of them are indifferentiable with the original analysis acting on the complete data.  The second simulation showcases the capability of our method in a more complicated problem that involves a hierarchical nonparametric prior. Our estimator has a clear advantage in accuracy over a naive estimator that imposes the prior in the original analysis on split data.

\section{A GWAS data set}\label{sec6}

The algorithm \eqref{ddsc} was designed to analyze data from GWAS studies. We focus on one particular study data, GIANT Height. Here we provide a concise introduction that paves the way to apply the method, more thorough information can be found in \cite{allen2010hundreds}. The trait variable height is collected for $N = 133,653$ individuals of recent European ancestry. After an initial screening, the number of single-nucleotide polymorphisms (SNPs) that are of interest is reduced to $n = 941,389$, of whom the regression coefficients $W_i$ and their associated standard error $\sigma_i$, $i = 1, \ldots, n$, in accordance to a simple linear regression are available.

Upon simple derivation the observed effect size $W_i$ is related to the true effect size $X_i$ through a measurement error model $W_i = X_i + U_i$, where $U_i \sim \mbox{N}(0,\sigma_{i\epsilon}^2/N)$ and $\sigma_{i\epsilon}^2$ equals the variance of regression error in the linear regression of height on the $i$th SNP. Because of the large number of individuals in the study, the variance of $U_i$ is well estimated by the standard error of the regression coefficient $W_i$ and thus is treated as known. By assuming all the true effect sizes $X_i$'s come from one distribution with $f(\cdot)$ as its density function, and acknowledging the fact that the observed effect sizes $W_i$ are symmetric with a majority near zero, it is reasonable to infer that $f(\cdot)$ is unimodal and symmetric at zero.  The natural question is how to estimate $f(\cdot)$.  It is apparrent that applying an efficient algorithm matters when one realizes that the number of SNPs is so large, which is typical for a GWAS data set.

The problem introduced has the same setup as our second simulation. We are interested in applying the divide and conquer algorithm \eqref{ddsc_dc} in this paper and compare the density estimators with those obtained by blindly using the original algorithm \eqref{ddsc} under the identical splitting of data. The same strategy as used in the simulation section is adopted for combining the individual density estimators. These two estimators are referred to as fPrior and naive correspondingly. We select $J = 50$ or $J = 200$ machines with each assigned effect sizes of around $20,000$ SNPs in the former case or around $5,000$ in the latter. As usual, a posterior mean estimator is chosen as the density estimator on individual machine.

Figure \ref{Fig2} displays the estimated densities for the true effect size $X$ under fPrior or naive. We display the peak and tail areas of these estimators separately, which are determined by a pre-specified cutoff $0.003$. Table \ref{tab3} summarizes the integrated absolute value of the difference (IAD) between the estimated densities under $J = 50$ and $J = 200$ for each method. IADs are also calculated separately for the peak region and the tail region.

One conclusion from Figure \ref{Fig2} and Table \ref{tab3} is that estimating the effect size density using the proposed prior is less sensitive to the total machines being used than that using the naive prior and thus we believe the proposed prior is advantageous, especially in estimating the tails.  We also observe that both priors are not very consistent in estimating the density around zero when different number of machines are selected reflecting the intrinsic difficulty in estimating the extremely small effect sizes. In conclusion, our method leads to a feasible solution in practice, especially when the focus is on detecting the larger effect sizes.
 
\begin{table}
\centering
\begin{tabular}{lccccccccccc}
\hline\hline
  Density Area &&&\multicolumn{3}{c}{$|x| > 0.003$}
  &&&\multicolumn{3}{c}{$|x| < 0.003$}\\
\hline
Estimator &&& fPrior & &naive &&& fPrior & &naive\\*[-.60em]
IAD ($50 - 200$)  &&& 0.003 & & 0.015 &&& 0.033 & &0.035  \\
\hline\hline
\end{tabular}
\caption{The integrated absolute value between the effect size density estimators under
  the choices of $J=50$ and $J = 200$, IAD ($50 - 200$), is
  compared separately for our method (fPrior) and for the naive prior
  (naive). The metric is calculated for two regions: the larger
  effect sizes ($|x| > 0.003$) and the extremely
small effect sizes ($|x| < 0.003$)}
\label{tab3}
\end{table}


\begin{figure}[htbp]
  \centering
  \includegraphics[height=0.45\linewidth, width=0.95\linewidth,
  trim=0.5cm 0.5cm 0.5cm 0.5cm, clip=TRUE]{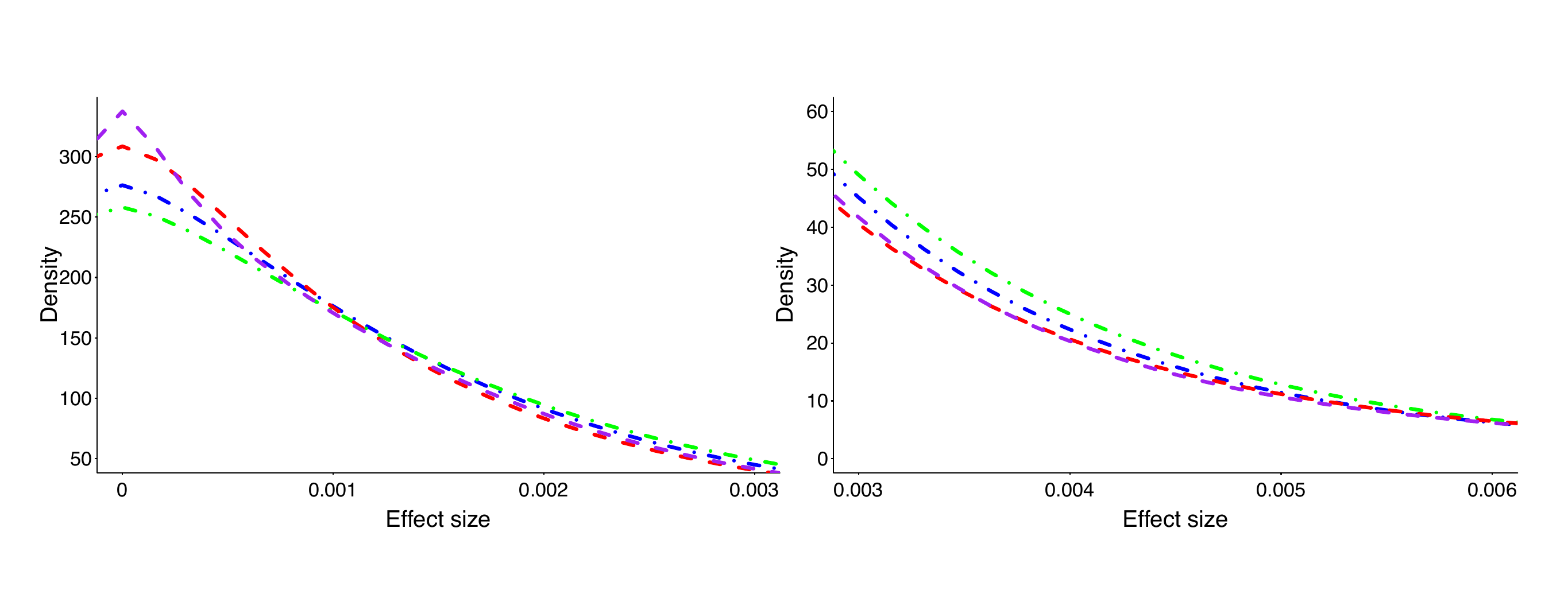}
  \caption{True effect size density estimators for
    the GIANT Height data set containing the observed effect sizes
    (associated with height) corresponding to $941,389$ SNPs separately. The peak
    region and tail region (on the positive values) are displayed in the left
    and right panel. Four density estimators are displayed: 
    fPrior using divided data (of equal size) on $50$
    machines (red dashed line), fPrior using divided data
    on $200$ machines (purple dashed line), the naive prior (original
    prior) using divided data on $50$ machines (blue dotted dashed
    line), the naive prior using divided data on $200$ machines (green dotted dashed
    line).}
\label{Fig2}
\end{figure}


\section{Discussions}\label{sec7}
We study a scalable Bayesian parametric or nonparametric 
density estimation method using a divide and conquer strategy, with a guaranteed optimal posterior
convergence. In addition, our numerical and real data results show the applicability of 
the method to a density deconvolution problem. There is an interest to
see how the idea can be used in an even broader context.

Our theoretical results indicate it is expected that the number of
machines can not be chosen to be too large compared to the total
sample size. In practice, how to select the total number of subsets is an open
problem.  

We are aware of a very nice theoretical result
\citep{szabo2017asymptotic} that investigates the optimal posterior
convergence regardless of the choice for the number of subsets.  They
consider the signal-in-noise model with a conjugate normal prior
which is regulated by the ``decay" rate of the true signal. The
authors cast a doubt about the existence of an adaptive version of
the method. One conjecture based on our theoretical results
and others is that a nonparametric prior is
crucial in possibly obtaining an adaptive convergence rate, together with a
control over the growth rate of the number of subsets.

\section*{Supplementary Material}

The online supplementary material includes detailed
algorithms for the selected examples that we used in the main paper,
the additional lemmas that slightly modifies the existing ones in the
literature for completeness and clearness.

\baselineskip=14pt

\section*{Acknowledgments}
Su is supported by a startup fund from College of Arts and Sciences,
University of Kentucky. The Authors are grateful to Yan Zhang of
Johns Hopkins University for directions to the online GWAS data set.

\clearpage\pagebreak\newpage
\newcommand{\Appendix}{\appendix\def\thesection{Appendix~\Alph{section}}\def\thesubsection{\Alph{section}.\arabic{subsection}}}
\section*{Appendix}
\begin{appendix}
\Appendix\renewcommand{\theequation}{A.\arabic{equation}}
\renewcommand{\thesubsection}{A.\arabic{subsection}}
\renewcommand{\thecondition}{A.\arabic{condition}}
\renewcommand{\theremark}{A.\arabic{remark}}
\renewcommand{\theproposition}{A.\arabic{proposition}}
\renewcommand{\thethm}{A.\arabic{thm}}
\renewcommand{\thelemma}{A.\arabic{lemma}}
\setcounter{equation}{0}
\setcounter{lemma}{0}
\baselineskip=18pt
\section*{Appendix}
\subsection{Proof of Theorem} \label{sec:app:1}

\begin{proof}
This is a proof of Theorem \ref{conv_rate_sub}. We will show that for
the subset posterior distributions $\Pi_m(\cdot  \mid  \sbf{X}_j)$ the same
contraction rate is achieved as if the original prior has been assigned on each machine. 

It is easy to argue that all $\Pi_m(\cdot  \mid  \sbf{X}_j)$ share the same asymptotic results,
it is sufficient to do for a single one. With a slight abuse of
notation, $\Pi_m(\cdot  \mid  \sbf{X}_m)$ denotes the posterior distribution on
$f(\cdot)$ in accordance to the adjusted prior, where
$\sbf{X}_m = (X_1, \ldots, X_m)$ corresponds to a sample of size $m$. Thus, 
\begin{eqnarray*}
  \Pi_m(\cdot  \mid  \sbf{X}_m) = \frac{\prod_{i=1}^m f(X_i) d\Pi_J(f)}{\int \prod_{i=1}^m f(X_i) d\Pi_J(f)}.
\end{eqnarray*}
Here $\Pi_J(f)$ is the prior given by \eqref{DPmixmvn_pprior}. 

We can follow the procedure in \cite{shen2013adaptive} which extends \cite{ghosal2007posterior} to derive the contraction rate for $\Pi_m(\cdot  \mid  \sbf{X}_j)$. The former leads to an adaptive rate assuming the true function is in a locally H\"{o}lder class. The difference lies
in the fraction prior that is used in this paper. The proof is built upon that of Theorem 1 in \cite{shen2013adaptive}. Here we aim to organize the outline of the proof and discuss the differences whenever necessary.  

Given the assumptions, a sieve space corresponding to $\epsilon_m$ is
constructed as $\mathscr{F}_m$. The procedure requires three major
steps:

\hspace{-0.5cm} 1. $\log \mathcal{N}(\epsilon_m, \mathscr{F}_m^c, h) \leq c_1 m \epsilon_m^2$ when $m
  \rightarrow \infty$, 

\hspace{-0.5cm} 2. $\Pi_J(\mathscr{F}_{m}^c) \leq c_3 \exp[{- (c_2 + 4) m \wt{\epsilon}_m^2}]$, 

\hspace{-0.5cm} 3. $\Pi_J(f: P_0\log(f/f_0) \leq \wt{\epsilon}_m^2, P_0\log^2(f/f_0) \leq
  \wt{\epsilon}_m^2) \geq c_4 \exp(- c_2 m \wt{\epsilon}_m^2)$.

We can show that all three steps above hold for $\mathscr{F}_m = \{f: f =
\int \phi_\sigma(x - z) P(dz) \mbox{ with } P = \sum_{h=1}^\infty
\pi_h \delta_{z_h}, z_h \in [-a_m, a_m]\, \mbox{ if } h \leq H_m;
\sum_{h > H_m} \pi_h < \epsilon_m; \sigma^2_m \leq \sigma^2 <
\sigma^2_m (1 + \epsilon_m^2)^{M_m}\}$, where $a_m^{a_1} = \sigma_m^{-2a_2} = M_m = m$, $H_m = \lfloor m \epsilon_m^2/(\log m)\rfloor$, $\epsilon_m = m^{-\beta/(2\beta + 1)} (\log m)^t$ and $\wt{\epsilon}_m = m^{-\beta/(2\beta + 1)} (\log m)^{t_0}$, $t - 1/2 > t_0 = 1 + (1/\tau)/(2 + 1/\beta)$. The detailed proof follows the proof of Theorem 1 in
\cite{shen2013adaptive} where we see that the effect of the
proposed prior leads to a contraction rate $\epsilon_m$ that is only a
factor $(\log m)^{1/2}$ larger than if using the original prior.

Since the same sieve space as \cite{shen2013adaptive} is adopted and a
larger $\epsilon_m$ (a factor of $\log m$), the first step on the entropy of the sieve space remains the same.

The second step will be discussed with more details. From the definition of sieve space $\mathscr{F}_m$, it can be easily argued that under prior \eqref{DPmixmvn_pprior}, $\mbox{DP}(M/J, G) \times \Pi_\sigma^{1/J}$, 
\begin{eqnarray}\label{prior_mass}
  (\mbox{DP}(M/J, G) \times \Pi_\sigma^{1/J}) (\mathscr{F}_m^c) &\leq& H_m G([-a_m, a_m]^c) + Pr({\textstyle \sum_{h > H_m}} \pi_h > \epsilon_m) \notag\\
  &+& \Pi_\sigma^{1/J}(\sigma^{-2} > \sigma_m^{-2}) + \Pi_\sigma^{1/J}(\sigma^{-2} \leq \sigma_m^{-2}(1 + \epsilon_m^2)^{-M_m}).
\end{eqnarray}

According to the assumptions on $G$, $\Pi_\sigma$ and $\mathscr{F}_m$, it can be shown that 
\begin{eqnarray*}
  H_m G([-a_m, a_m]^c) &\leq& b_1 m \epsilon_m^2(\log m)^{-1} \exp(-C_1 m),\\
  \Pi_\sigma^{1/J}(\sigma^{-2} > \sigma_m^{-2}) &\leq& b_2^{1/J} \exp(-C_2 m/J), \\
  \Pi_\sigma^{1/J}(\sigma^{-2} \leq \sigma_m^{-2}(1 + \epsilon_m^2)^{-M_m}) &\leq& b_3^{1/J} m^{a_3/(a_2 J)} (1 + \epsilon_m^2)^{-m/J} \\
  &\asymp& \exp(-C_3 m \epsilon_m^2/J),
\end{eqnarray*}
in addition, the second term in \eqref{prior_mass}, the upper bound
for $Pr({\textstyle  \sum_{h > H_m}} \pi_h > \epsilon_m )$, can be
found using stick breaking representation for $\{\pi_h, h \geq 1\}$,
that is, $\pi_h = V_h \prod_{i < h}
(1 - V_i)$, $\{V_h , h \geq 1\}$ are independent beta-distributed
random variables with parameter $1$ and $M/J$. Then we can show 
\begin{align*}
Pr({\textstyle  \sum_{h > H_m}} \pi_h > \epsilon_m) &= Pr\{{\textstyle
  \prod_{i=1}^{H_m}} (1 - V_i) > \epsilon_m\} = Pr\{-{\textstyle
  \sum_{i = 1}^{H_m}} \log(1 - V_i) < \log(1/\epsilon_m)\} \\
 &\leq \frac{\{-(M/J) \log \epsilon_m\}^{H_m}}{\Gamma(H_m + 1)} \leq
        \bigg(\frac{e M}{H_m J} \log \frac{1}{\epsilon_m}\bigg)^{H_m}
        \asymp \exp\{- C_4 m \epsilon_m^2 \log m\}.
\end{align*}
The last two inequalities follow from $-\sum_{i = 1}^{H_m} \log(1 -
V_i)$ is a Gamma random variable with parameter $H_m$ and $M/J$ and
Stirling's formula.

These upper bounds together with \eqref{prior_mass} yield
$(\mbox{DP}(M/J, G) \times \Pi_\sigma^{1/J}) (\mathscr{F}_m^c)
\lesssim \exp(-C m \epsilon_m^2/J)$ for some constant $C$. Since $m =
n/J, J \asymp \log n$ and $\epsilon_m^2 > \wt{\epsilon}_m^2 \log m$, we conclude that $(\mbox{DP}(M/J, G) \times \Pi_\sigma^{1/J}) (\mathscr{F}_m^c) \leq c_3 \exp\{-(c_2 + 4)m \wt{\epsilon}_m^2\}$ for some constant $c_3$ and any constant $c_2$, which will be chosen as the constant in step three.  

The third step is termed ``Prior thickness result" in
\cite{shen2013adaptive} and the lower bound therein is built for prior
\eqref{DPmixmvn}, $\mbox{DP}(M, G) \times \Pi_\sigma$, on a prior set
$\mathcal{P}_{\sigma_m} \times \mathcal{S}_{\sigma_m}$ (in the one dimensional
case $\mathcal{S}_{\sigma_m} = \{\sigma: \sigma^{-2} \in [\sigma_m^{-2},
\sigma_m^{-2}(1 + \sigma_m^{2\beta})]\}$) with
$\sigma_m^\beta = \wt{\epsilon}_m \{\log^{-1} (1/\wt{\epsilon}_m)\}$,   
in recognition of Lemma 10 in \cite{ghosal2007posterior} and the condition of
$\Pi_\sigma$. Thus it can be easily verified the lower bound holds for
the same $\wt{\epsilon}_m$ under prior \eqref{DPmixmvn_pprior},
$\mbox{DP}(M/J, G) \times \Pi_\sigma^{1/J}$ on the prior set
$\mathcal{P}_{\sigma_m} \times \mathcal{S}_{\sigma_m}$, as long as the same lower
bound of Lemma 10 and that of $\Pi_\sigma$ can be achieved. It is
easy to show that $\Pi_\sigma^{1/J}(\mathcal{S}_{\sigma_m}) \geq C_4
\exp[-c_4 \wt{\epsilon}^{-1/\beta}\{\log(1/\wt{\epsilon}_m)\}^{2 +
  1/\tau + 1/\beta}]$ under Condition (C2) and (C4).
We can also show that the same lower bound of Lemma 10 in
\cite{ghosal2007posterior} holds with
a slight modification on the condition of Dirichlet parameters. For
readers' interest it is stated as Lemma S.1
and provided in the Supplementary material.
 
\end{proof}

\bibliographystyle{biomAbhra}
\bibliography{Bayesian_divide_conquer}

\end{appendix}

\clearpage\pagebreak\newpage
\pagestyle{fancy}
\fancyhf{}
\rhead{\bfseries\thepage}
\lhead{\bfseries NOT FOR PUBLICATION SUPPLEMENTARY MATERIAL}
\begin{center}
{\LARGE{\bf Supplementary Material to\\ {\it Divide and Conquer algorithm of Bayesian Density Estimation}}}
\end{center}

\baselineskip=12pt

\vskip 2mm
\begin{center}
Ya Su \\
Department of Statistics, University of Kentucky, Lexington, KY 40536-0082, U.S.A., ya.su@uky.edu\\
\end{center}

\setcounter{figure}{0}
\setcounter{equation}{0}
\setcounter{page}{1}
\setcounter{table}{0}
\setcounter{section}{0}
\renewcommand{\thefigure}{S.\arabic{figure}}
\renewcommand{\theequation}{S.\arabic{equation}}
\renewcommand{\thesection}{S.\arabic{section}}
\renewcommand{\thesubsection}{S.\arabic{section}.\arabic{subsection}}
\renewcommand{\thepage}{S.\arabic{page}}
\renewcommand{\thetable}{S.\arabic{table}}
\renewcommand{\thelemma}{S.\arabic{lemma}}

\section{Algorithms for example models}\label{sec.S1}

The algorithms for three example models in the main body of the paper are displayed. For ease of notation, let $\mathbf{\Omega}_{-\zeta}$ be all variables in $\mathbf{\Omega}$ but excluding $\zeta$. 

\subsection{Finite mixtures of Normal prior}\label{sec.S1.1}

The Gibbs sampling algorithm is given below. 

Denote $X_{ji}$ the $i$th sample distributed to subset $j$, $Z_{ji}$
the component indicator variable where $X_{ji}$ pertains to
that is, $X_{ji}|Z_{ji} = k \sim \mbox{N}_p(\sbf{\mu}_k, \sbf{\Sigma}_k)$, $n_{jk} = \sum_{i=1}^m I_{(Z_{ji} = k)}$, $\widebar{X}_{jk} =
n_{jk}^{-1} \sum_{i=1}^m X_{ji} I_{(Z_{ji} = k)}$  and $V_{jk} =
\sum_{Z_{ji} = k}(X_{ji} - \widebar{X}_{jk})(X_{ji} - \widebar{X}_{jk})^T$ corresponding to
the number of samples, sample mean and scaled sample covariance matrix
belonging to subset $j$ and component $k$ for $j = 1, \ldots, J$, $k =
1,\ldots, K$. Under the new prior for each subset, the conditional posterior
distributions of the variables are  
\begin{eqnarray*}
  P(Z_{ji} = k|\sbf{\Omega}_{-Z_i}) &\propto& \pi_k \mbox{N}_p(X_{ji}; \sbf{\mu}_k,
                                        \sbf{\Sigma}_k) \\ 
  \sbf{\Sigma}_k | \sbf{\Omega}_{-\{\sbf{\Sigma}_k,\sbf{\mu}_k\}} &\sim& \mbox{IW}\left(n_{jk} +
                                               \frac{\nu + 1}{J} - \frac{(p+1)(J-1)}{J}, V_{jk}
                                               + \frac{(l J)^{-1} n_{jk}}{(l J)^{-1} +
                                 n_{jk}} \widebar{X}_{jk}
                                               \widebar{X}_{jk}^T + \frac{1}{J} \sbf{S} \right)\\
  \sbf{\mu}_k | \sbf{\Omega}_{-\sbf{\mu}_k} &\sim&  \mbox{N}_p\left(\frac{n_{jk}}{(l J)^{-1} +
                                 n_{jk}} \widebar{X}_{jk}, \frac{1}{(l J)^{-1} +
                                 n_{jk}} \sbf{\Sigma}_k\right)    \\
  \sbf{\pi} | \sbf{\Omega}_{-\sbf{\pi}} &\sim&  \text{Dir}(n_{j1} + \alpha_1
                                         J^{-1}, \ldots, n_{jK} + \alpha_K J^{-1}).
\end{eqnarray*}

\subsection{Shape constraint density deconvolution}\label{sec.S1.5}

To ease computation, we approximate the Dirichlet process mixture prior with a finite mixture of Gamma distributions with $K$ components, with a specific Dirichlet prior on the mixture probabilities \citep{ishwaran2002exact}. Specifically, our hierarchical Bayes model for subsequent implementations is as follows. Let $i$ denote the index for subject, and $k$ be the index for the $k$th component, for all $i = 1, \ldots, n$, $k = 1, \ldots, K$. Let $t > 1$ denote a fixed constant. Then,
\begin{eqnarray*}
  &&(W_i|X_i) \sim \text{Normal}(X_i, \sigma_i^2); \enspace
  (X_i|\theta_i) \sim \text{Unif}(-\theta_i, \theta_i); \enspace
  (\theta_i|Z_i = k, \alpha_k, \beta_k) \sim \text{Ga}(\alpha_k,
  \beta_k); \\
  &&P(Z_i = k|p_1, \ldots, p_K) = p_k; \enspace
  (\alpha_k|\lambda, t) \sim \text{Expon}(\lambda; t, \infty); \enspace
  (\beta_k|\Xi_1, \Xi_2) \sim \text{Ga}(\Xi_1, \Xi_2); \\
  &&(p_1 \ldots, p_K) \sim \text{Dir}(m/K, \ldots, m/K),
\end{eqnarray*}
where $\text{Expon}(\lambda; \ell, u)$ denotes an exponential distribution with parameter $\lambda$ truncated at $(\ell, u)$. The truncation of $\alpha_k$ at some $t > 1$ makes the density of $X$ be finite at zero. The set of hyperparameters is $(\lambda, t, \Xi_1, \Xi_2, K, m)$.

For $k = 1, \ldots, K$, let $r_k = \sum_i I_{(Z_i = k)}$ be
the total number of individuals that fall into group $k$ and $s_k = \sum_i \theta_i I_{(Z_i = k)}$ be the summation of the $\theta_i$ from the $k$th group. To sample from the posterior distribution of $\Omega$, we use a Gibbs sampler for all parameters other than the $\alpha_k$, combined with a Metropolis-Hastings within Gibbs for the $\alpha_k$. The posterior full-conditional distributions are
\begin{eqnarray*}\label{eq:posterior}
  (X_i | \mathbf{\Omega}_{-X_i}) &\sim& \mbox{N}(W_i, \sigma_i^2; -\theta_i, \theta_i);\\
  (\theta_i | \mathbf{\Omega}_{-\theta_i}) &\sim& \text{Ga}(\alpha_{Z_i} - 1,
                                         \beta_{Z_i}; |X_i|, \infty);\\
  P(Z_i = k|\mathbf{\Omega}_{-Z_i}) &\propto& \Gamma(\alpha_k)^{-1} p_k (\beta_k \theta_i)^{\alpha_k}
                                     \exp(-\beta_k \theta_i);
  \\
  (p_1, \ldots, p_K | \mathbf{\Omega}_{- \{p_1, \ldots, p_K\}}) &\sim& \text{Dir}(m/K + r_1, \ldots, m/K + r_K);\\
  (\beta_k | \mathbf{\Omega}_{-\beta_k}) &\sim& \text{Ga}(\Xi_1 + \alpha_k r_k, \Xi_2 + s_k);
  \\
  (\alpha_k | \mathbf{\Omega}_{-{\alpha_k}}) &\propto&
                                                \Gamma(\alpha_k)^{-r_k} \exp\{-
                                                \alpha_k(\lambda -
                                                r_k \log \beta_k -
                                                \sum_i
                                                \log(\theta_i) I_{(Z_i = k)})\}.
\end{eqnarray*}
The symbol $\mbox{N}(\mu, \sigma^2; \ell, u)$ denotes a Normal distribution with parameters $(\mu, \sigma^2)$ truncated at $(\ell, u)$, while $\text{Ga}(\alpha, \beta; \ell, u)$ corresponds to a Gamma distribution with parameters $(\alpha, \beta)$ truncated at $(\ell, u)$. Since the posterior distribution of $\alpha_k$ does not belong to a standard family, we implement a Metropolis-Hastings algorithm within the Gibbs sampler to update the $\alpha_k$. We use a Gamma proposal distribution; specifically, $\wt{\alpha}_k \sim \text{Ga}(2, 2/\alpha_k; t, \infty)$, and we accept the proposed $\wt{\alpha}_k$ or keep the original $\alpha_k$ according to the general Metropolis-Hastings rule. The proposal distribution is truncated to reflect the prior assumption on $\alpha_k$.

\section{Additional Lemmas}\label{sec.S1.2}
The following lemma is a straightforward extension to Lemma 10 in \cite{ghosal2007posterior}. It turns out that the same conclusion holds for a Dirichlet-distributed random variable when the sum of its associated parameters have limit zero. 
\begin{lemma}\label{ghosal10}{\rm
  For $(p_1, \ldots, p_N)$ be an arbitrary point in the
  $N$-dimensional unit simplex and let $(X_1, \ldots, X_N)$ be
  Dirichlet distributed with parameter $(\alpha_1, \ldots, \alpha_N)$
  with $\alpha_j \leq 1$ and $\sum_{j=1}^N \alpha_j = m_N$. Suppose $\lim_{N
    \rightarrow \infty} m_N = 0$. Then for every $\epsilon^b < a
  \alpha_j$ and $\epsilon_N \leq 1$, there exists constants $c$ and
  $C$ that depend only on $a$ and $b$ such that 
  \begin{eqnarray*}
    \Pr\bigg(\sum_{j=1}^N|X_j - p_j| \leq 2\epsilon, \min_{1 \leq j \leq N}
    X_j \geq \epsilon^2/2\bigg) \geq C \exp(-c N \log \epsilon^{-1}).
  \end{eqnarray*}
}
\end{lemma}

The proof of the above lemma can follow exactly the lines of Lemma
10 in \cite{ghosal2007posterior} and thus is omitted.

\end{document}